\documentclass[aps,prl,twocolumn,superscriptaddress,nofootinbib]{revtex4-2}

\usepackage[utf8]{inputenc}
\usepackage{textgreek}
\usepackage[T1]{fontenc}
\usepackage{amsmath,amssymb,mathtools,bm}
\usepackage{amsthm}
\usepackage{microtype}
\usepackage[colorlinks=true,linkcolor=blue,citecolor=blue,urlcolor=blue]{hyperref}
\usepackage{enumitem}
\usepackage{color,xcolor}

\theoremstyle{plain}
\newtheorem{theorem}{Theorem}

\theoremstyle{definition}

\newcommand{\Sgen}{S_{\text{gen}}}

\DeclareMathOperator{\Tr}{Tr}

\newcommand{\Rmnum}[1]{\expandafter\@slowromancap\romannumeral #1@}

\newcommand{\be}{\begin{equation}}
\newcommand{\ee}{\end{equation}}

\newcommand{\I}{\infty}

\renewcommand{\a}{\alpha}	
\renewcommand{\d}{\delta}	

\renewcommand{\l}{\lambda}	

\newcommand{\q}{\theta}	
	
\begin{document}

\title{No off-diagonal quantum focusing for R\'enyi divergences}

\author{Tanay Kibe}
\email{tanay.kibe@ib.edu.ar}
\affiliation{Instituto Balseiro, Centro 
At{\'o}mico Bariloche, S.C. de Bariloche, 8400, 
R{\'i}o Negro, Argentina}
\author{Pratik Roy}
\email{roy.pratik92@gmail.com}
\affiliation{Institute of Mathematics, University of Warsaw, ul. Banacha 2, 02-097 Warsaw, Poland}

\date{\today}
\begin{abstract}
The quantum focusing conjecture (QFC) is a mathematical expression of the idea that semiclassical gravity remains universally attractive. 
In its non-gravitational  limit, the diagonal inequality is the quantum null energy condition (QNEC), which requires convexity of relative entropy with respect to the vacuum under null deformations. 
R\'enyi divergences generalize relative entropy, and positive results for the diagonal R\'enyi QNEC raise the question of whether they also satisfy the off-diagonal inequalities required by a full focusing principle. We show that classical correlations between distinct null generators obstruct this extension for any non-affine R\'enyi-type divergence satisfying data processing, tensor additivity, and matched classical–quantum conditioning.
\end{abstract}

\maketitle

\section{Introduction}
The classical focusing theorem is a precise mathematical formulation of the intuition that gravity is universally attractive. 
Its generalization to the semiclassical regime is the quantum focusing conjecture (QFC) \cite{Bousso:2015mna}.
In the nongravitational limit, the QFC reduces to a diagonal and an off-diagonal inequality on the relative entropy of an excited state with respect to the vacuum. 
This raises the natural question of whether these inequalities survive when the relative entropy is replaced by a R\'enyi divergence\footnote{These are R\'enyi relative entropies and should not be confused with the R\'enyi entropies $S_\alpha(
\rho)=\frac{1}{\alpha-1}\log \Tr(\rho^\alpha)$, which generalize the entanglement entropy.}. Here we prove that the off-diagonal inequality rules out all non-affine R\'enyi generalisations, uniquely selecting relative entropy among the currently known operationally relevant families.
Our result isolates the properties that a distinguishability measure must satisfy for a universal focusing principle to hold.

In spacetimes satisfying the null curvature condition, the classical focusing theorem asserts that the expansion, $\q$, of a hypersurface-orthogonal null congruence is non-increasing,
\begin{equation}\label{eq:classical-focusing}
    \frac{d\q}{d\l}\leq0, \qquad \q=\nabla_ak^a = \lim_{A\to0}\frac1A\frac{dA}{d\l}.
\end{equation}
Here, $\l$ is an affine parameter along the congruence with tangent vector $k^a$, and $A$ is the local infinitesimal area element of the congruence.  

Let $\sigma$ be a codimension-2 spatial surface that splits a Cauchy surface into two parts, and let $N$ be its orthogonal null hypersurface (Fig.~\ref{fig:activation}).
Divide $N$ into pencils of infinitesimal width around its null generators and write $\lambda(y)$ for the affine parameter along the pencil through $y$ on $\sigma$. 
Bousso et al.~\cite{Bousso:2015mna} defined the quantum expansion $\Theta$ in terms of the generalized entropy $S_{\rm gen}$ associated with the cut. The QFC then asserts that
\begin{equation}\label{eq:qfc}
    \frac{\delta\Theta(y_1)}{\delta \l(y_2)}\leq0, \qquad \Theta(y_1)=\frac{4G_N\hbar}{\sqrt{h(y_1)}}\frac{\delta \Sgen}{\delta \l(y_1)},
\end{equation}
where  $G_N$ is Newton's constant, $\delta/\delta\l(y_i)$ varies the null cut at transverse location $y_i$, and $\sqrt{h}$ is the area element of the deformed cut.

The off-diagonal QFC, with $y_1\neq y_2$ in \eqref{eq:qfc}, was argued to follow from strong subadditivity (SSA) of entanglement entropy \cite{Bousso:2015mna}, whereas the diagonal part is more non-trivial. 
In the non-gravitational limit, i.e. $G_N\to0$, the diagonal QFC was termed the quantum null energy condition (QNEC).
Diagonal QNEC can be phrased as the statement that the relative entropy $D(\Psi\|\Omega)$ of any state $\Psi$ with respect to the vacuum $\Omega$ has non-negative second null shape variations at coincident points
\begin{equation}\label{eq:qnec}
    \frac{\d^2 D(\Psi\|\Omega)}{\d\l(y)^2}\geq0.
\end{equation} 
In this form, the QNEC has been proven in algebraic quantum field theory (QFT) \cite{Ceyhan:2018zfg,Hollands:2025glm} (see also \cite{Koeller:2015qmn,Balakrishnan:2017bjg} for other proofs) using the fact that the algebras associated to null deformed regions are half-sided modular inclusions (HSMIs) \cite{Borchers:1991xk,Wiesbrock:1992mg,Borchers:1995zg,Borchers1996,Araki:2005we}.

Following the identification of $\Sgen$ with the relative entropy \cite{Witten:2021unn,Chandrasekaran:2022eqq}, it has recently been proposed
\cite{Chandrasekaran:2026pnc} that the QNEC proof of \cite{Ceyhan:2018zfg} can be
imported to establish the diagonal QFC at leading non-trivial order in perturbative quantum gravity on Killing horizon backgrounds.

Motivated by QNEC, \cite{Lashkari:2018nsl} conjectured that a similar statement as \eqref{eq:qnec} should hold also for sandwiched R\'{e}nyi divergence (SRD), $\widetilde D_\a(\Psi\|\Omega)$, a R\'{e}nyi generalization of relative entropy defined first in \cite{Muller-Lennert:2013liu, Wilde:2013bdg}, and generalized to arbitrary von Neumann algebras in \cite{Berta:2016vnw, Jencova:2016tqz,Jencova:2017txf}. 
The diagonal R\'{e}nyi QNEC (RQNEC) was proved for free bosonic theories for $\a\in[1,\I)$, with counterexamples for $\a<1$ \cite{Moosa:2020jwt}, and a free fermion proof followed in \cite{Roy:2022yzm}. 
We recently proved RQNEC for all integer $\a\geq2$ using HSMI, assuming only that $\Psi$ has finite SRD \cite{Kibe:2026wsg}. 
However, these positive diagonal, non-gravitational results constrain only the second variations at a single null ray. A full focusing principle also requires inequalities for variations at two distinct rays. 
It is therefore natural to ask if there are also valid R\'{e}nyi generalizations of the off-diagonal QNEC and of full quantum focusing, not least because QNEC was first derived from the QFC. 

This was addressed partially in~\cite{Moosa:2020jwt}, where the authors showed that strong superadditivity of SRD fails for general Markov reference states. 
However, the QFT vacuum has more structure than Markovianity alone, so this general failure does not settle whether the off diagonal inequality holds for the vacuum.

In this Letter, we show that for a broad axiomatic class of R\'{e}nyi divergences, off-diagonal
generalizations of QNEC can be violated even when the reference state is the QFT vacuum. Since any proposed generalization of quantum focusing must reduce to a valid QFT inequality in the non-gravitational $G_N\to0$ limit, it follows that no divergence in this class can satisfy a universal off-diagonal focusing inequality.
Our result applies to any divergence $D_\gamma$ satisfying the axioms in the treatment of \cite{Muller-Lennert:2013liu, Tomamichel:2015gtd}. 
For the standard Petz, sandwiched, and \(\alpha\)-\(z\) \cite{Audenaert:2015npv,Kato:2023aro,Kato:2023hlj,Hiai:2024qve} R\'enyi families, Umegaki relative entropy  is uniquely selected. 

In null-quantized free field theory, distinct transverse pencils are independent \cite{Wall:2011hj}. We prepare two pencils in states conditioned on a shared classical record, which is stored in spectator pencils that are left undeformed. Away from the affine ($\gamma=0$) limit, the mixed second null shape variation acquires a classical covariance term, and this has the wrong sign for the off-diagonal focusing inequality.
Our construction thus shows that focusing constrains the underlying distinguishability measure beyond what data processing or the diagonal R\'enyi QNEC require.

\section{R\'{e}nyi divergences}\label{sec:physical_divergences}

Let \(D_\gamma(\psi\|\omega)\) be a one-parameter family of divergences
on normalized density matrices, with \(\gamma=0\) denoting the affine
limiting member.  For \(\gamma\neq0\) define the moment
\begin{equation}
    Q_\gamma(\psi\|\omega)=\exp[\gamma D_\gamma(\psi\|\omega)] .
    \label{eq:moment}
\end{equation}
We require that $D_\gamma$ is faithful,
\[
    D_\gamma(\psi\|\omega)\ge0,\qquad
    D_\gamma(\psi\|\omega)=0 \iff \psi=\omega ,
\]
satisfies the data processing inequality (DPI), is additive on tensor products
\eqref{eq:tensor}, and obeys the matched classical-quantum (cq) conditioning rules
\eqref{eq:flag-renyi} and \eqref{eq:flag-affine} for
\(\gamma\neq0\) and \(\gamma=0\), respectively.
These conditions are contained within or follow from the axioms of \cite{Muller-Lennert:2013liu,Tomamichel:2015gtd}. 
Standard Petz, sandwiched, and more generally $\alpha$-$z$ R\'enyi divergences satisfy these axioms whenever finite (see End Matter).

Tensor additivity means that
\begin{equation}
D_\gamma\!\left(\bigotimes_i\psi_i\Big\|\bigotimes_i\omega_i\right)=\sum_iD_\gamma(\psi_i\|\omega_i).
\label{eq:tensor}
\end{equation}
The classical--quantum (cq) conditioning rule applies to direct sums of classical sectors (see e.g.~\cite{Tomamichel:2015gtd}). Let \(\{P_\eta\}\) be mutually orthogonal
projections on a finite register system, and let
\(\chi_\eta\) be normalized density matrices of the register, with
\(\chi_\eta=P_\eta\chi_\eta P_\eta\). For a probability
distribution \(\mu_\eta\), define
\[
\psi=\sum_\eta\mu_\eta\psi_\eta\otimes\chi_\eta,
\qquad
\omega=\sum_\eta\mu_\eta\omega_\eta\otimes\chi_\eta,
\]
with the sum understood as a direct sum over orthogonal
\(\eta\)-sectors. The register itself carries no distinguishability; it only records which branch was prepared. For \(\gamma\neq0\), the cq conditioning rule is
\begin{align}
D_\gamma(\psi\|\omega)
&=
\frac1\gamma
\log\sum_\eta\mu_\eta
\exp\!\left[
\gamma D_\gamma(\psi_\eta\|\omega_\eta)
\right],
\label{eq:flag-renyi}
\end{align}
and for $\gamma=0$, it is
\begin{equation}
D_0(\psi\|\omega)
=
\sum_\eta\mu_\eta D_0(\psi_\eta\|\omega_\eta).
\label{eq:flag-affine}
\end{equation}
We refer to the $\gamma=0$ point as the affine limit, where the log-sum-exp collapses to a sum.  For the SRD $\widetilde D_{\alpha}(\rho\|\sigma)$ and the \(\alpha\)-\(z\) R\'{e}nyi divergence $\widetilde D_{\alpha,z}(\rho\|\sigma)$, the limit $\alpha-1=\gamma\to0$ gives Umegaki relative entropy.

\section{Null quantization setup}

Consider a free quantum field theory on
$d>2$ dimensional Minkowski spacetime with metric
\begin{equation}
    ds^2=-dx^+dx^-+d\vec y^{\,2},
\end{equation}
where $x^\pm$ are null coordinates and $\vec y$ denotes the $d-2$ transverse coordinates. A function $V(\vec y)\geq0$ specifies a cut of the null hypersurface $N:x^-=0$, with $V(\vec y)=\lambda(\vec y)$ on the pencil at $\vec y$ (Fig.~\ref{fig:activation}). In null quantization, the degrees of freedom to the future of a cut
$V(\vec y)$ factorize over transverse pencils,
\begin{equation}
    \mathcal H_V\simeq \bigotimes_j \mathcal H_{j,V_j},
    \qquad V_j\equiv V(\vec y_j),
    \label{eq:null-factorization}
\end{equation}
up to the usual zero-mode qualifications. See e.g. \cite{Wall:2011hj,Bousso:2015wca,Malik:2019dpg,Moosa:2020jwt,Roy:2022yzm} for applications of null quantization to proofs of the generalized second law, QNEC and R\'{e}nyi QNEC. 
After introducing a UV regulator on each pencil, the vacuum density
matrix correspondingly factorizes as
\begin{equation}
    \omega_V=\bigotimes_j \omega_j(V_j).
    \label{eq:vac-factorization}
\end{equation}
A universal off-diagonal focusing inequality must already hold in the regulated theory with finitely many pencils in excited states. 

\begin{figure*}[t]
\centering
\includegraphics[width=\textwidth]{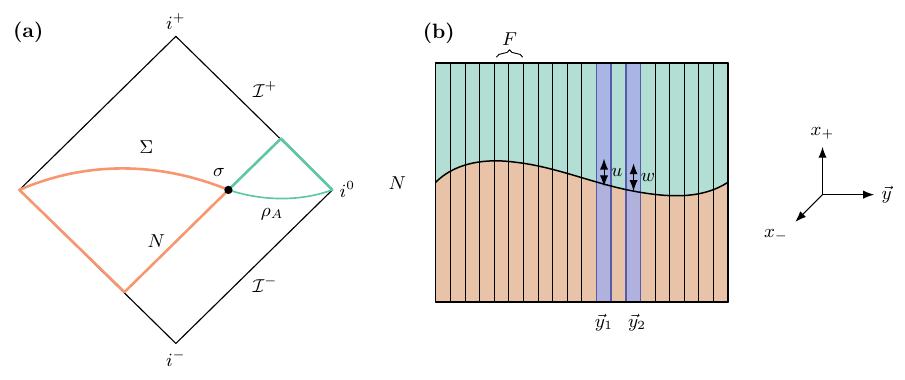}
\caption{(a) Penrose diagram of Minkowski spacetime with a Cauchy surface $\Sigma$ that is split into two by the co-dimension two surface $\sigma$. $N$ is a null hypersurface orthogonal to $\sigma$.  The state $\rho_A$ on the subregion $A$ of $\Sigma$ can be equivalently described on the green part of $N$ together with the green part of future null infinity $\mathcal I^+$. (b) The null hypersurface $N$ is divided into null pencils of infinitesimal width around its null generators. A function $V(\vec y)$ of the transverse coordinates splits $N$ into two parts (green and orange).
The two pencils of interest are colored blue and labeled by $\vec{y}_{1,2}$, with null deformation parameters $u,w$. $F$ indicates a subset of the spectator pencils that serve as the classical register.}
\label{fig:activation}
\end{figure*}

We deform two active pencils, located at transverse positions
$\vec y_1,\vec y_2$, by
\begin{equation}
    V_{u,w}(\vec z)
    =
    V(\vec z)
    +u\,v_{\vec y_1}(\vec z)
    +w\,v_{\vec y_2}(\vec z),
    \label{eq:two-pencil-deformation}
\end{equation}
with $v_{\vec y_i}(\vec z) = 1$ in a neighborhood with infinitesimal area around the pencil at $\vec y_i$, and zero otherwise.
Denote the state $\psi$ and the vacuum density matrix $\omega$ restricted to the future of the deformed cut \eqref{eq:two-pencil-deformation} by $\psi(u,w),\omega(u,w)$, and write
\begin{equation}
    D_\gamma(u,w)
    := 
    D_\gamma \! \left(\psi(u,w)\middle \| \omega(u,w)\right).
\end{equation}
For two cut values $u,w\geq 0$, define the off-diagonal finite difference
\begin{equation}
    \Delta_{12}D_\gamma
    :=
    D_\gamma(w,w)+D_\gamma(u,u)
    -D_\gamma(u,w)-D_\gamma(w,u).
    \label{eq:finite-difference}
\end{equation}
An infinitesimal off-diagonal focusing inequality for
$D_\gamma(\psi\|\omega)$ requires this quantity to be non-negative
for all sufficiently small rectangles $(u,w)$.  Conversely, a negative finite
difference in the regulated theory implies, under usual continuity
assumptions on the cut dependence, a negative mixed variation
somewhere in the rectangle.

We choose pencils $1,2$ to be the active pencils, and the undeformed cut to be $V(\vec z)=0$.  For each active pencil,
let $\psi_i(c)$ be the reduced density matrix of a chosen one-pencil
excitation on the portion of the pencil to the future of the cut
$x^+=c$, and let $\omega_i(c)$ be the corresponding vacuum density
matrix.  Define
\begin{align}
    d_i(c)
    &:=
    D_\gamma\!\left(\psi_i(c)\middle\|\omega_i(c)\right),
    \label{eq:di-defn}
    \\
    q_i(c)
    &:=
    Q_\gamma\!\left(\psi_i(c)\middle\|\omega_i(c)\right)
    =
    e^{\gamma d_i(c)}
    \qquad (\gamma\neq0).
    \label{eq:qi-defn}
\end{align}
If pencil $i$ is not activated, its state is $\omega_i(c)$ and its
contribution to the divergence is zero.

Choose a finite set $F$ of undeformed spectator
pencils at transverse positions away from $\vec y_1,\vec y_2$ and
outside the supports of the deformation profiles
$v_{\vec y_1},v_{\vec y_2}$.  Denote the remaining undeformed pencils
by $R$.  In the regulated theory
\begin{equation}
    \omega(u,w)
    =
    \omega_1(u)\otimes \omega_2(w)\otimes \omega_F\otimes\omega_R,
    \label{eq:vac-active-spectator}
\end{equation}
where $\omega_F$ and $\omega_R$ are independent of $u,w$.

Choose mutually orthogonal projections
$\{P_\eta\}_{\eta\in E}$ on the spectator-pencil Hilbert space,
where $E\subseteq\{0,1\}^2$, such that
\begin{equation}
    \sum_{\eta\in E}P_\eta=\mathbf 1_F,
    \qquad
    P_\eta P_{\eta'}=\delta_{\eta\eta'}P_\eta,
    \qquad
    [P_\eta,\omega_F]=0.
\end{equation}
Set
\begin{equation}
    \mu_\eta=\operatorname{Tr}(P_\eta\omega_F),
    \qquad
    \chi_\eta=\mu_\eta^{-1}P_\eta\omega_F P_\eta,
    \label{eq:vacuum-blocks}
\end{equation}
omitting any zero-weight block. The spectator vacuum then has the block decomposition
\begin{equation}
    \omega_F=\sum_{\eta\in E}\mu_\eta\chi_\eta.
    \label{eq:spectator-vacuum-block}
\end{equation}
Thus $\eta$ labels orthogonal sectors of the undeformed spectator degrees of freedom. The register is part of the QFT itself.

For a binary activation pattern $\eta=(\eta_1,\eta_2)$ define
\begin{equation}
    \tau_{i,\eta_i}(c)
    =
    \begin{cases}
        \psi_i(c), & \eta_i=1,\\
        \omega_i(c), & \eta_i=0.
    \end{cases}
    \label{eq:tau-defn}
\end{equation}
The counterexample below uses the regulated QFT density matrix
\begin{align}
    \psi_\mu(u,w)
    =
    \sum_{\eta\in E}\mu_\eta
    \Big[
        \tau_{1,\eta_1}(u)
        \otimes
        \tau_{2,\eta_2}(w)
        \otimes
        \chi_\eta
    \Big]
    \otimes \omega_R .
    \label{eq:pencil-excited-state}
\end{align}
The reference state in the divergence is the QFT vacuum, decomposed as
\begin{align}
    \omega(u,w)
    =
    \sum_{\eta\in E}\mu_\eta
    \Big[
        \omega_1(u)
        \otimes
        \omega_2(w)
        \otimes
        \chi_\eta
    \Big]
    \otimes\omega_R .
    \label{eq:vacuum-reference-block}
\end{align}
We note that the spectator density matrix of $\psi_\mu(u,w)$ is exactly $\omega_F$, the
same as in the vacuum.  Hence the spectator pencils carry no
distinguishability by themselves, and they only store which active-pencil
preparation is correlated with each vacuum block.

Because the same block weights and the same spectator states
$\chi_\eta$ occur in both arguments, the matched cq conditioning rule \eqref{eq:flag-renyi}
applies to the block decomposition. For
$\gamma\neq0$,
\begin{align}
    &Q_\gamma\!\left(\psi_\mu(u,w)\middle\|\omega(u,w)\right)
    =
    \sum_{\eta\in E}\mu_\eta
    q_1(u)^{\eta_1}q_2(w)^{\eta_2},
    \label{eq:active-pencil-Q}
    \\
    &D_\gamma(u,w)
    =
    {1\over\gamma}
    \log
    \sum_{\eta\in E}\mu_\eta
    \exp\!\left[
        \gamma\bigl(\eta_1d_1(u)+\eta_2d_2(w)\bigr)
    \right].
    \label{eq:active-pencil-partition-function}
\end{align}
In the affine limit, with
$p_i:=\sum_{\eta\in E}\mu_\eta\eta_i$, we obtain
\begin{equation}
    D_0(u,w)=p_1d_1(u)+p_2d_2(w).
    \label{eq:relative-flag-affine}
\end{equation}

Writing
\begin{equation}
    r=e^{\gamma d_1(u)},\qquad
    s=e^{\gamma d_2(w)},
\end{equation}
and extending $\mu_\eta$ to all four activation patterns,
we have
\begin{equation}
    D_\gamma(u,w)={1\over\gamma}\log Z(r,s),
    \label{eq:two-pencil-Z}
\end{equation}
with
\begin{equation}
    Z(r,s)=
    \mu_{00}+\mu_{10}r+\mu_{01}s+\mu_{11}rs.
\end{equation}
A direct calculation gives 
\begin{equation}
    {\partial^2\log Z\over
    \partial\log r\,\partial\log s}
    =
    {rs(\mu_{00}\mu_{11}-\mu_{10}\mu_{01})\over Z^2}=:\operatorname{Cov}_{\widetilde\mu}(\eta_1,\eta_2).
    \label{eq:activation-association}
\end{equation}
The right-hand side is the covariance of the two activation variables with respect to the reweighted probabilities $\widetilde\mu_\eta=\mu_\eta r^{\eta_1}s^{\eta_2}/Z$. Where the one-pencil profiles are differentiable, the chain rule gives
\[
    \partial_u\partial_w D_\gamma
    =\gamma\,d_1'(u)d_2'(w)\,
    \operatorname{Cov}_{\widetilde\mu}(\eta_1,\eta_2).
\]
Data processing makes both $d_i$ non-increasing, but does not fix the sign of the covariance. Where both one-pencil derivatives are nonzero, negative covariance gives the wrong sign for $\gamma>0$, and positive covariance does so for $\gamma<0$. At $\gamma=0$, the affine expression \eqref{eq:relative-flag-affine} is separable in the cuts and this contribution vanishes. The finite-difference proof below does not require differentiability.

\section{No-go theorem}

We call a one-pencil excitation a seed.  For a fixed divergence
$D_\gamma$, a seed is finite and non-rigid if there are two cut
values $u,w$ on the same pencil such that
\begin{equation}
    d_u:=D_\gamma\!\left(\psi(u)\middle\|\omega(u)\right),
    \qquad
    d_w:=D_\gamma\!\left(\psi(w)\middle\|\omega(w)\right),
\end{equation}
are finite and $d_u\neq d_w$.  The theorem below assumes such a seed exists
for the value of $\gamma$ under consideration.

Such seeds are abundant. By data processing under restriction to later
cuts, \(d(c)=D_\gamma(\psi(c)\|\omega(c))\) is non-increasing in $c$ as the
cut moves to the future, whereas saturation of the DPI is non-generic (see \cite{Petz:1986tvy,Jencova:2017let}). More concretely, choose a bounded local
unitary \(U\) supported between two cuts \(u<w\), and set
\(\psi(c)={\rm tr}_{(-\I,c)}(U^\dagger|\Omega\rangle\langle\Omega| U )\). For cuts later than the
support of \(U\), locality gives \(\psi(c>w)=\omega(c)\), hence
\(d(c>w)=0\). For an earlier cut \(u<c<w\), the excitation is visible, so
\(\psi(u<c<w)\neq\omega(c)\), and faithfulness of the divergence gives
\(d(c)>0\), whenever the quantity is finite. Thus a finite non-rigid seed exists whenever the chosen divergence is
finite on such bounded local excitations. 

\begin{theorem}\label{thm:no-go}
Fix a R\'enyi divergence $D_\gamma$ with $\gamma\neq0$ satisfying DPI, faithfulness, and Eqs.~\eqref{eq:tensor}, \eqref{eq:flag-renyi} and \eqref{eq:flag-affine}.
Assume that a finite non-rigid one-pencil seed exists.  Then there is
a density matrix $\psi$ of the regulated QFT, with reference state the
QFT vacuum density matrix $\omega$, such that
\begin{equation}
    \Delta_{12}D_\gamma(\psi\|\omega)<0.
\end{equation}
Hence no universal off-diagonal
focusing inequality for $D_\gamma(\psi\|\omega)$ can hold for
$\gamma\neq0$ in this class.
\end{theorem}

\begin{proof}
Take two identical copies of the non-rigid seed, one on each active
pencil, so that the one-pencil divergence is the same function
$d(c)$ on pencils $1$ and $2$.  Choose a nontrivial spectator-vacuum
block decomposition with two blocks.  Equivalently, choose a
projection $P$ on the spectator pencils, commuting with $\omega_F$,
such that
\begin{equation}
    0<p:=\operatorname{Tr}(P\omega_F)<1.
\end{equation}
The complementary block has weight $1-p$. 

First suppose $\gamma>0$.  Use the two spectator blocks to implement
the exclusive activation pattern
\begin{equation}
    \mu_{10}=p,\qquad
    \mu_{01}=1-p,\qquad
    \mu_{00}=\mu_{11}=0.
    \label{eq:exclusive-law}
\end{equation}
Then, from \eqref{eq:two-pencil-Z}
\begin{align}
    &D_\gamma(u,u)=d_u, \quad 
    D_\gamma(w,w)=d_w,
    \\
    &D_\gamma(u,w)
    =
    {1\over\gamma}
    \log\!\left[
        p e^{\gamma d_u}
        +(1-p)e^{\gamma d_w}
    \right],
    \\
    &D_\gamma(w,u)
    =
    {1\over\gamma}
    \log\!\left[
        p e^{\gamma d_w}
        +(1-p)e^{\gamma d_u}
    \right].
\end{align}
Let
\begin{equation}
    r=e^{\gamma d_u},
    \qquad
    s=e^{\gamma d_w}.
\end{equation}
Since $d_u\neq d_w$ and $\gamma\neq0$, we have $r\neq s$.  Moreover,
\begin{align}
\begin{aligned}
    \left(p r+(1-p)s\right)
 \left(p s+(1-p)r\right)-rs=\\
p(1-p)(r-s)^2>0.
\end{aligned}
\end{align}
Therefore
\begin{multline}
    \gamma\,\Delta_{12}D_\gamma
    =\\
    \qquad\log(rs)
    -
    \log\!\left(
        [p r+(1-p)s][p s+(1-p)r]
    \right)
    <0.
\end{multline}
Since $\gamma>0$, this gives $\Delta_{12}D_\gamma<0$.

Now suppose $\gamma<0$.  Use the two spectator blocks to implement
the positively correlated activation pattern
\begin{equation}
    \mu_{00}=p,\qquad
    \mu_{11}=1-p,\qquad
    \mu_{10}=\mu_{01}=0.
    \label{eq:correlated-law}
\end{equation}
Then
\begin{equation}
    D_\gamma(u,w)=g_\gamma(d(u)+d(w)),
\end{equation}
where
\begin{equation}
    g_\gamma(t)
    :=
    {1\over\gamma}
    \log\!\left[
        p+(1-p)e^{\gamma t}
    \right].
\end{equation}
A direct derivative gives
\begin{equation}
    g_\gamma''(t)
    =
    \gamma\,
    {p(1-p)e^{\gamma t}
    \over
    \left[p+(1-p)e^{\gamma t}\right]^2}
    <0,
    \qquad \gamma<0.
\end{equation}
Thus $g_\gamma$ is strictly concave.  Since $d_u\neq d_w$,
\begin{align}
    \Delta_{12}D_\gamma
    &=
    g_\gamma(2d_w)+g_\gamma(2d_u)
    -2g_\gamma(d_u+d_w) <0.
    \label{eq:negative-gamma-violation}
\end{align}
This excludes every finite nonzero $\gamma$.

By contrast, at $\gamma=0$, Eq.~\eqref{eq:relative-flag-affine} gives
\begin{equation}
    D_0(u,w)=p_1d(u)+p_2d(w),
\end{equation}
which is separable in the two cuts.  Hence
\begin{equation}
    \Delta_{12}D_0=0
\end{equation}
for the same spectator-block constructions.
\end{proof}
The two preparation patterns realize negative and positive covariance, respectively, in \eqref{eq:activation-association}.

The same construction excludes the max-relative entropy endpoint ($\lim_{\a\to \I} \widetilde D_\a$) of the SRD,
provided a finite non-rigid $D_{\max}$ seed exists.
In this case,
\begin{equation}
        \Delta_{12}D_{\max}
        =
        d_1+d_2-2\max\{d_1,d_2\}
        =
        -|d_1-d_2|<0.
        \label{eq:dmax-violation}
\end{equation}
The violation therefore persists at this endpoint.

The SRD satisfies (DPI) and the hypothesis of Theorem~\ref{thm:no-go} in the range $[1/2,1)\cup(1,\I)$, with affine limit $\lim_{\a\to 1} \widetilde D_\a=D$. An explicit coherent-state SRD seed was constructed in~\cite{KibeRoy:2026FixedRayQNEC} for the chiral $U(1)$ half-line current:
\begin{equation}
    d_\alpha(c)=\frac{\pi q^2g^2}{12(g+c)^2}\,\alpha(\alpha+1),
    \qquad \alpha>1,
    \label{eq:coherent-seed-srd}
\end{equation}
where $g>0$, $q\neq0$, and $0\leq c<g/4$. This is finite and strictly
decreasing. Inserting it into the counterexample construction
\eqref{eq:exclusive-law}, with $\gamma=\alpha-1$, gives a negative
$\Delta_{12}\widetilde D_\alpha$ for distinct cuts (see End Matter). In the same setting, explicit seeds for relative entropy can be found in \cite{Longo:2018obd}.

\section{Discussion}

Off-diagonal focusing constrains how distinguishability responds to correlated preparations on different null generators, and our construction shows that data processing and tensor additivity do not supply this constraint. 
Through the non-linear cq conditioning rule, classical correlations between two preparations produce a mixed second variation of the wrong sign.
The reference state throughout is the vacuum, so the additional structure of the QFT vacuum does not remove the obstruction.
Our result therefore excludes off-diagonal focusing for every non-affine R\'enyi divergence satisfying our assumptions.

Combined with existing positive results, this yields an operational characterization of relative entropy. 
Non-negativity of its mixed null variations follows from strong subadditivity \cite{Bousso:2015mna}, its diagonal variations obey QNEC \cite{Ceyhan:2018zfg,Hollands:2025glm}, and the full quantum focusing conjecture has been argued for at leading order in perturbative quantum gravity on Killing horizon backgrounds \cite{Chandrasekaran:2026pnc}. 
Our theorem shows that within the R\'{e}nyi class no other non-affine divergence can enter such a statement. 
For the standard R\'enyi families with known operational interpretations, there is exactly one candidate entropic ingredient for a universal statement that quantum gravity is attractive: the quantum expansion built from generalized entropy, Eq.~\eqref{eq:qfc}.

Our no-go leaves the diagonal question open. The recent proof of the diagonal (sandwiched) R\'enyi QNEC for integer orders \cite{Kibe:2026wsg} establishes the $G_N\to0$ limit of a possible diagonal R\'enyi focusing statement, and our counterexample cannot obstruct it. Whether a diagonal sandwiched-R\'enyi QFC holds once gravitational backreaction is included remains an interesting open problem.

It is also interesting to ask whether other divergences can support useful null-energy or focusing statements. For example, Belavkin--Staszewski (BS) relative entropy \cite{Belavkin1982CalgebraicGO} and geometric relative entropies \cite{Mosonyi:2022nsp} satisfy all of our axioms and are not excluded by our counter-example since they are affine under matched-classical-quantum conditioning. 
However, a focusing statement based on their R\'enyi generalizations (see for example \cite{Bergh:2021hzf}) satisfying our assumptions is excluded. Null-energy inequalities for the BS and geometric relative entropies remain to be investigated.

\textbf{Acknowledgements.} We thank Stefan Hollands for helpful discussions on quantum focusing. T.K. is supported by a Simons Foundation fellowship through the Targeted Grant to Instituto Balseiro. The work of P.R. has been supported by the Polish National Science Centre through Sonata grant (2022/47/D/ST2/02058).


\begin{center}
{\large\bf End Matter}
\end{center}

\setcounter{equation}{0}
\renewcommand{\theequation}{A\arabic{equation}}
\renewcommand{\theHequation}{supp.\arabic{equation}}
\setcounter{theorem}{0}
\setcounter{proposition}{0}
\setcounter{corollary}{0}
\setcounter{lemma}{0}

\section{Classical--quantum conditioning \\ for \texorpdfstring{$\alpha$-$z$}{alpha-z} R\'enyi divergence}

We derive the cq conditioning rule used in the main text.

Let \(F\) be a finite-dimensional register system with mutually orthogonal
projections \(\{P_\eta\}_{\eta\in E}\),
\begin{equation}
    P_\eta P_{\eta'}=\delta_{\eta\eta'}P_\eta,
    \qquad
    \sum_{\eta\in E}P_\eta=\mathbf 1_F .
\end{equation}
Let \(\chi_\eta\) be normalized density matrices supported on these
orthogonal sectors,
\begin{equation}
    \chi_\eta=P_\eta\chi_\eta P_\eta,
    \qquad
    \Tr\chi_\eta=1 .
\end{equation}
For normalized density matrices \(\psi_\eta,\omega_\eta\) on a quantum system \(A\),
define
\begin{align}
\begin{aligned}
    \psi_{AF}
    &=
    \sum_{\eta\in E}p_\eta\,
    \psi_\eta\otimes\chi_\eta
    =
    \bigoplus_{\eta\in E}
    p_\eta(\psi_\eta\otimes\chi_\eta),
    \\
    \omega_{AF}
    &=
    \sum_{\eta\in E}q_\eta\,
    \omega_\eta\otimes\chi_\eta
    =
    \bigoplus_{\eta\in E}
    q_\eta(\omega_\eta\otimes\chi_\eta).
\end{aligned}
    \label{eq:S-cqstates}
\end{align}
The sums are direct sums over orthogonal \(\eta\)-sectors.  We assume the
usual support and finite-domain conventions for the divergences, and omit zero-weight branches.
The matched case used in the main text is \(p_\eta=q_\eta=\mu_\eta\).  The
identical register factors \(\chi_\eta\) contribute zero divergence in each sector and provide the classical label.

For \(\alpha,z>0\), \(\alpha\neq1\), the finite-dimensional
\(\alpha\)-\(z\) moment is
\begin{equation}
    Q_{\alpha,z}(\rho\|\sigma)
    =
    \Tr\left[
    \left(
    \sigma^{\frac{1-\alpha}{2z}}
    \rho^{\alpha/z}
    \sigma^{\frac{1-\alpha}{2z}}
    \right)^z
    \right],
    \label{eq:S-alphaz}
\end{equation}
and the corresponding divergence is defined as
\begin{equation}
    D_{\alpha,z}(\rho\|\sigma)
    =
    \frac{1}{\alpha-1}\log Q_{\alpha,z}(\rho\|\sigma).
\end{equation}
The Petz and sandwiched divergences are recovered by the choices
\(z=1\) and \(z=\alpha\), respectively.

First, \(Q_{\alpha,z}\) is multiplicative on tensor products.  Functional
calculus gives
\begin{equation}
\begin{aligned}
    &(\psi_1\otimes\psi_2)^{\alpha/z}
    =
    \psi_1^{\alpha/z}\otimes\psi_2^{\alpha/z},
    \\
    &(\omega_1\otimes\omega_2)^{\frac{1-\alpha}{2z}}
    =
    \omega_1^{\frac{1-\alpha}{2z}}
    \otimes
    \omega_2^{\frac{1-\alpha}{2z}},
\end{aligned}
\end{equation}
and the trace factorizes.  Hence
\begin{equation}
    Q_{\alpha,z}(\psi_1\otimes\psi_2\|
    \omega_1\otimes\omega_2)
    =
    Q_{\alpha,z}(\psi_1\|\omega_1)
    Q_{\alpha,z}(\psi_2\|\omega_2).
    \label{eq:S-tensormoment}
\end{equation}
Equivalently, \(D_{\alpha,z}\) is additive on tensor products.

We now evaluate \(Q_{\alpha,z}\) on the block states
\eqref{eq:S-cqstates}.  Since powers and functional calculus preserve
direct sums,
\begin{align}
&\omega_{AF}^{\frac{1-\alpha}{2z}}
 \psi_{AF}^{\alpha/z}
 \omega_{AF}^{\frac{1-\alpha}{2z}}
\nonumber\\
&\quad =
\bigoplus_{\eta\in E}
(q_\eta\omega_\eta\otimes\chi_\eta)^{\frac{1-\alpha}{2z}}
(p_\eta\psi_\eta\otimes\chi_\eta)^{\alpha/z}
(q_\eta\omega_\eta\otimes\chi_\eta)^{\frac{1-\alpha}{2z}}
\nonumber\\
&\quad =
\bigoplus_{\eta\in E}
p_\eta^{\alpha/z}q_\eta^{(1-\alpha)/z}
\left(
\omega_\eta^{\frac{1-\alpha}{2z}}
\psi_\eta^{\alpha/z}
\omega_\eta^{\frac{1-\alpha}{2z}}
\right)
\otimes
\chi_\eta^{1/z}.
\label{eq:S-blockinside}
\end{align}
Here powers of \(\chi_\eta\) are understood on its support.  Raising a
block diagonal operator to the \(z\)-th power acts block by block, and
\[
    \left(A_\eta\otimes \chi_\eta^{1/z}\right)^z
    =
    A_\eta^z\otimes\chi_\eta .
\]
Taking the trace and using \(\Tr\chi_\eta=1\) gives
\begin{equation}
    Q_{\alpha,z}(\psi_{AF}\|\omega_{AF})
    =
    \sum_{\eta\in E}
    p_\eta^\alpha q_\eta^{1-\alpha}
    Q_{\alpha,z}(\psi_\eta\|\omega_\eta).
    \label{eq:S-cqmoment-general}
\end{equation}
Thus
\begin{align}\label{eq:S-cqD-general}
    D_{\alpha,z}&(\psi_{AF}\|\omega_{AF})\\
    =&
    \frac{1}{\alpha-1}
    \log
    \sum_{\eta\in E}
    p_\eta^\alpha q_\eta^{1-\alpha}
    \exp\left[
    (\alpha-1)D_{\alpha,z}(\psi_\eta\|\omega_\eta)
    \right].        \nonumber
\end{align}
In the matched case \(p_\eta=q_\eta=\mu_\eta\), we get
\begin{equation}
    Q_{\alpha,z}(\psi_{AF}\|\omega_{AF})
    =
    \sum_{\eta\in E}
    \mu_\eta
    Q_{\alpha,z}(\psi_\eta\|\omega_\eta),
    \label{eq:S-flagmoment}
\end{equation}
or equivalently
\begin{align}\label{eq:S-flagD}
     D_{\alpha,z}&(\psi_{AF}\|\omega_{AF}) \\
    &=
    \frac{1}{\alpha-1}
    \log
    \sum_{\eta\in E}
    \mu_\eta
    \exp\!\left[
    (\alpha-1)D_{\alpha,z}(\psi_\eta\|\omega_\eta)
    \right].    \nonumber
\end{align}
This is the matched cq conditioning rule used in the main text.  If, for fixed $z$,
\(\gamma=\alpha-1\) and \(D_\gamma\equiv D_{\alpha,z}\), then
\eqref{eq:S-flagD} is the log-sum-exp law (for $\gamma \neq0$)
\begin{equation}
    D_{\gamma}(\psi_{AF}\|\omega_{AF})
    =
    \frac{1}{\gamma}
    \log
    \sum_{\eta\in E}
    \mu_\eta
    \exp\!\left[
    \gamma D_{\gamma}(\psi_\eta\|\omega_\eta)
    \right].
    \label{eq:S-flagD-gamma}
\end{equation}

This formula is the block version of R\'enyi's general-mean property.
The strictly monotone transform
\begin{equation}
    g_\alpha(t)=\exp[(\alpha-1)t],
    \qquad
    Q_{\alpha,z}(\rho\|\sigma)=g_\alpha(D_{\alpha,z}(\rho\|\sigma)),
\end{equation}
is affine over matched classical alternatives:
\begin{equation}
    g_\alpha\!\left(D_{\alpha,z}(\psi_{AF}\|\omega_{AF})\right)
    =
    \sum_{\eta\in E}
    \mu_\eta\,
    g_\alpha\!\left(
    D_{\alpha,z}(\psi_\eta\|\omega_\eta)
    \right).
\end{equation}
Thus the R\'enyi moment \(Q_{\alpha,z}\), rather than
\(D_{\alpha,z}\) itself, averages over matched classical sectors. 

Setting $z=\alpha$, and taking the limit \(\alpha\to1\) gives the ordinary cq chain rule for
relative entropy.  For the general block states \eqref{eq:S-cqstates} in the matched $p_\eta=q_\eta=\mu_\eta$ case,
\begin{equation}
    D_0(\psi_{AF}\|\omega_{AF})
    =
    \sum_{\eta\in E}
    \mu_\eta D_0(\psi_\eta\|\omega_\eta).
    \label{eq:S-relative-cq-general}
\end{equation}
Thus relative entropy is affine under matched classical conditioning,
whereas nonzero R\'enyi parameters produce the log-sum-exp law.

\section{An explicit coherent-state violation}
\label{sec:coherent-state-example}

We use the coherent excitation of the chiral $U(1)$ current constructed in~\cite{KibeRoy:2026FixedRayQNEC}. For $g>0$ and real
$q\neq0$, its smooth global current one-point function agrees with
\begin{equation}
    j_\psi(x)=\langle J(x)\rangle_\psi
    =\frac{qg}{(x+g)^2},\qquad x>-g/2.
    \label{eq:S-coherent-current}
\end{equation}
The profile is smoothly completed for $x<-g/2$, with a compensating pulse
chosen so that its integral vanishes, avoiding zero-mode issues. This completion lies outside all the
half-lines considered below. The SRD of the restricted state with respect
to the vacuum is \eqref{eq:coherent-seed-srd}.

Fix $\alpha>1$ and two cuts $0\leq u<w<g/4$, and write
$C_\alpha:=\pi q^2g^2\alpha(\alpha+1)/12$. Their one-pencil divergences
differ by
\begin{equation}
    \delta_\alpha
    :=d_\alpha(u)-d_\alpha(w)
    =C_\alpha
    \left[\frac{1}{(g+u)^2}-\frac{1}{(g+w)^2}\right]>0.
    \label{eq:S-coherent-seed-difference}
\end{equation}
Take identical copies on the two active pencils and correlate their
preparations with the spectator record according to
\eqref{eq:exclusive-law}, with $0<p<1$. Set
$\gamma=\alpha-1$, $r=e^{\gamma d_\alpha(u)}$, and
$s=e^{\gamma d_\alpha(w)}$. Equations~\eqref{eq:finite-difference} and
\eqref{eq:active-pencil-partition-function} give
\begin{multline}
    \Delta_{12}\widetilde D_\alpha
    =\\
    -\frac1\gamma
    \log\frac{[pr+(1-p)s][ps+(1-p)r]}{rs}.
    \label{eq:S-coherent-ratio}
\end{multline}
Using
\begin{multline}
    [pr+(1-p)s][ps+(1-p)r]\\
    =rs+p(1-p)(r-s)^2
    \label{eq:S-coherent-ratio-identity}
\end{multline}
and $(r-s)^2/(rs)=4\sinh^2(\gamma\delta_\alpha/2)$, we obtain
\begin{multline}
    \Delta_{12}\widetilde D_\alpha
    =\\
    -\frac1{\alpha-1}
    \log\!\left[
    1+4p(1-p)\sinh^2\!\left(\frac{\alpha-1}{2}\delta_\alpha\right)
    \right]<0.
    \label{eq:S-coherent-violation}
\end{multline}
Since $p(1-p)>0$ and $\delta_\alpha>0$, the logarithm is strictly positive.
For fixed $\alpha,q,g,u,w,p$, the violation is finite and contains no ultraviolet
cutoff-dependent factor.

\bibliographystyle{apsrev4-1}
\bibliography{refs}

@article{Petz:1986tvy,
    author = "Petz, Denes",
    title = "{Sufficient subalgebras and the relative entropy of states of a von Neumann algebra}",
    doi = "10.1007/BF01212345",
    journal = "Commun. Math. Phys.",
    volume = "105",
    number = "1",
    pages = "123--131",
    year = "1986"
}

@article{Borchers:1991xk,
    author = "Borchers, H. J.",
    title = "{The CPT theorem in two-dimensional theories of local observables}",
    reportNumber = "GOET-TP-91-01",
    doi = "10.1007/BF02099011",
    journal = "Commun. Math. Phys.",
    volume = "143",
    pages = "315--332",
    year = "1992"
}

@article{Wiesbrock:1992mg,
    author = "Wiesbrock, H. W.",
    title = "{Half sided modular inclusions of von Neumann algebras}",
    reportNumber = "SFB-288-17",
    doi = "10.1007/BF02098019",
    journal = "Commun. Math. Phys.",
    volume = "157",
    pages = "83--92",
    year = "1993",
    note = "[Erratum: Commun.Math.Phys. 184, 683--685 (1997)]"
}

@article{Borchers1996,
  author  = {Borchers, Hans-Jürgen},
  title   = {Half-sided modular inclusion and the construction of the Poincaré group},
  journal = {Commun. Math. Phys.},
  year    = {1996},
  volume  = {179},
  number  = {3},
  pages   = {703--723},
  doi     = {10.1007/BF02100104}}

@article{Borchers:1995zg,
    author = "Borchers, H. J.",
    title = "{On the use of modular groups in quantum field theory}",
    journal = "Ann. Inst. H. Poincare Phys. Theor.",
    volume = "63",
    pages = "331--382",
    year = "1995"
}

@article{Araki:2005we,
    author = "Araki, H. and Zsido, L.",
    title = "{Extension of the structure theorem of Borchers and its application to half-sided modular inclusions}",
    eprint = "math/0412061",
    archivePrefix = "arXiv",
    doi = "10.1142/S0129055X05002388",
    journal = "Rev. Math. Phys.",
    volume = "17",
    pages = "491--543",
    year = "2005"
}

@article{Wall:2011hj,
    author = "Wall, Aron C.",
    title = "{A proof of the generalized second law for rapidly changing fields and arbitrary horizon slices}",
    eprint = "1105.3445",
    archivePrefix = "arXiv",
    primaryClass = "gr-qc",
    doi = "10.1103/PhysRevD.85.104049",
    journal = "Phys. Rev. D",
    volume = "85",
    pages = "104049",
    year = "2012",
    note = "[Erratum: Phys.Rev.D 87, 069904 (2013)]"
}

@article{Audenaert:2015npv,
    author = "Audenaert, Koenraad M. R. and Datta, Nilanjana",
    title = "{alpha-z-relative Renyi entropies}",
    eprint = "1310.7178",
    archivePrefix = "arXiv",
    primaryClass = "quant-ph",
    doi = "10.1063/1.4906367",
    journal = "J. Math. Phys.",
    volume = "56",
    pages = "022202",
    year = "2015"
}

@book{Tomamichel:2015gtd,
    author = "Tomamichel, Marco",
    title = "{Quantum Information Processing with Finite Resources. Mathematical Foundations}",
    eprint = "1504.00233",
    archivePrefix = "arXiv",
    primaryClass = "quant-ph",
    doi = "10.1007/978-3-319-21891-5",
    isbn = "978-3-319-21890-8, 978-3-319-21891-5",
    publisher = "Springer",
    series = "SpringerBriefs in Mathematical Physics",
    volume = "5",
    year = "2016"
}

@article{Belavkin1982CalgebraicGO,
  title={C*-algebraic generalization of relative entropy and entropy},
  author={Viacheslav P. Belavkin and Przemyslaw Staszewski},
  journal={Annales De L Institut Henri Poincare-physique Theorique},
  year={1982},
  volume={37},
  pages={51-58},
  url={https://api.semanticscholar.org/CorpusID:123727205}
}

@article{Bousso:2015mna,
    author = "Bousso, Raphael and Fisher, Zachary and Leichenauer, Stefan and Wall, Aron C.",
    title = "{Quantum focusing conjecture}",
    eprint = "1506.02669",
    archivePrefix = "arXiv",
    primaryClass = "hep-th",
    doi = "10.1103/PhysRevD.93.064044",
    journal = "Phys. Rev. D",
    volume = "93",
    number = "6",
    pages = "064044",
    year = "2016"
}

@article{Bousso:2015wca,
    author = "Bousso, Raphael and Fisher, Zachary and Koeller, Jason and Leichenauer, Stefan and Wall, Aron C.",
    title = "{Proof of the Quantum Null Energy Condition}",
    eprint = "1509.02542",
    archivePrefix = "arXiv",
    primaryClass = "hep-th",
    doi = "10.1103/PhysRevD.93.024017",
    journal = "Phys. Rev. D",
    volume = "93",
    number = "2",
    pages = "024017",
    year = "2016"
}

@article{Koeller:2015qmn,
    author = "Koeller, Jason and Leichenauer, Stefan",
    title = "{Holographic Proof of the Quantum Null Energy Condition}",
    eprint = "1512.06109",
    archivePrefix = "arXiv",
    primaryClass = "hep-th",
    doi = "10.1103/PhysRevD.94.024026",
    journal = "Phys. Rev. D",
    volume = "94",
    number = "2",
    pages = "024026",
    year = "2016"
}

@article{Berta:2016vnw,
    author = "Berta, Mario and Scholz, Volkher B. and Tomamichel, Marco",
    title = "{R\'enyi Divergences as Weighted Non-commutative Vector-Valued $L_p$ -Spaces}",
    eprint = "1608.05317",
    archivePrefix = "arXiv",
    primaryClass = "math-ph",
    doi = "10.1007/s00023-018-0670-x",
    journal = "Annales Henri Poincare",
    volume = "19",
    number = "6",
    pages = "1843--1867",
    year = "2018"
}

@article{Balakrishnan:2017bjg,
	archiveprefix = {arXiv},
	author = {Balakrishnan, Srivatsan and Faulkner, Thomas and Khandker, Zuhair U. and Wang, Huajia},
	doi = {10.1007/JHEP09(2019)020},
	eprint = {1706.09432},
	journal = {JHEP},
	pages = {020},
	primaryclass = {hep-th},
	title = {{A General Proof of the Quantum Null Energy Condition}},
	volume = {09},
	year = {2019}
}

@article{Lashkari:2018nsl,
    author = "Lashkari, Nima",
    title = "{Constraining Quantum Fields using Modular Theory}",
    eprint = "1810.09306",
    archivePrefix = "arXiv",
    primaryClass = "hep-th",
    doi = "10.1007/JHEP01(2019)059",
    journal = "JHEP",
    volume = "01",
    pages = "059",
    year = "2019"
}

@article{Ceyhan:2018zfg,
	archiveprefix = {arXiv},
	author = {Ceyhan, Fikret and Faulkner, Thomas},
	doi = {10.1007/s00220-020-03751-y},
	eprint = {1812.04683},
	journal = {Commun. Math. Phys.},
	number = {2},
	pages = {999--1045},
	primaryclass = {hep-th},
	title = {{Recovering the QNEC from the ANEC}},
	volume = {377},
	year = {2020}
}

@article{Jencova:2016tqz,
    author = "Jencova, Anna",
    title = "{R{\'e}nyi relative entropies and noncommutative $L_p$-spaces}",
    eprint = "1609.08462",
    archivePrefix = "arXiv",
    primaryClass = "quant-ph",
    doi = "10.1007/s00023-018-0683-5",
    journal = "Annales Henri Poincare",
    volume = "19",
    pages = "2513",
    year = "2018"
}

@article{Jencova:2017txf,
    author = "Jen\v{c}ov\'a, Anna",
    title = "{R\'enyi relative entropies and noncommutative $L_p$-spaces II}",
    eprint = "1707.00047",
    archivePrefix = "arXiv",
    primaryClass = "quant-ph",
    doi = "10.1007/s00023-021-01074-9",
    journal = "Annales Henri Poincare",
    volume = "22",
    pages = "3235--3254",
    year = "2021"
}

@misc{KibeRoy:2026FixedRayQNEC,
  author        = {Kibe, Tanay and Roy, Pratik},
  title         = {Quasi-local form for {$\alpha$--$z$} {R\'enyi QNEC} from fixed-ray escorts},
  year          = {2026},
  eprint        = {2609.04016},
  archivePrefix = {arXiv},
  primaryClass  = {hep-th},
  url           = {https://arxiv.org/abs/2609.04016}
}

@article{Malik:2019dpg,
	archiveprefix = {arXiv},
	author = {Malik, Taha A. and Lopez-Mobilia, Rafael},
	doi = {10.1103/PhysRevD.101.066028},
	eprint = {1910.07594},
	journal = {Phys. Rev. D},
	number = {6},
	pages = {066028},
	primaryclass = {hep-th},
	title = {{Proof of the quantum null energy condition for free fermionic field theories}},
	volume = {101},
	year = {2020}
}

@article{Moosa:2020jwt,
    author = "Moosa, Mudassir and Rath, Pratik and Su, Vincent Paul",
    title = "{A R\'enyi quantum null energy condition: proof for free field theories}",
    eprint = "2007.15025",
    archivePrefix = "arXiv",
    primaryClass = "hep-th",
    doi = "10.1007/JHEP01(2021)064",
    journal = "JHEP",
    volume = "01",
    pages = "064",
    year = "2021"
}

@article{Bergh:2021hzf,
    author = "Bergh, Bjarne and Salzmann, Robert and Datta, Nilanjana",
    title = "{The \textalpha{}$\to 1$ Limit of the Sharp Quantum R\'enyi Divergence}",
    eprint = "2102.06576",
    archivePrefix = "arXiv",
    primaryClass = "quant-ph",
    doi = "10.1063/5.0049791",
    journal = "J. Math. Phys.",
    volume = "62",
    pages = "092205",
    year = "2021"
}

@article{Roy:2022yzm,
    author = "Roy, Pratik",
    title = "{Proof of the R\'enyi quantum null energy condition for free fermions}",
    eprint = "2212.02331",
    archivePrefix = "arXiv",
    primaryClass = "hep-th",
    doi = "10.1103/PhysRevD.108.045010",
    journal = "Phys. Rev. D",
    volume = "108",
    number = "4",
    pages = "045010",
    year = "2023"
}

@article{Kato:2023hlj,
    author = "Kato, Shinya",
    title = "{On {\ensuremath{\alpha}}-z-R{\'e}nyi divergence in the von Neumann algebra setting}",
    eprint = "2311.01748",
    archivePrefix = "arXiv",
    primaryClass = "math.OA",
    doi = "10.1063/5.0186552",
    journal = "J. Math. Phys.",
    volume = "65",
    number = "4",
    pages = "042202",
    year = "2024"
}

@article{Wilde:2013bdg,
    author = "Wilde, Mark M. and Winter, Andreas and Yang, Dong",
    title = "{Strong Converse for the Classical Capacity of Entanglement-Breaking and Hadamard Channels via a Sandwiched Renyi Relative Entropy}",
    eprint = "1306.1586",
    archivePrefix = "arXiv",
    primaryClass = "quant-ph",
    doi = "10.1007/s00220-014-2122-x",
    journal = "Commun. Math. Phys.",
    volume = "331",
    number = "2",
    pages = "593--622",
    year = "2014"
}

@article{Chandrasekaran:2022eqq,
    author = "Chandrasekaran, Venkatesa and Penington, Geoff and Witten, Edward",
    title = "{Large N algebras and generalized entropy}",
    eprint = "2209.10454",
    archivePrefix = "arXiv",
    primaryClass = "hep-th",
    doi = "10.1007/JHEP04(2023)009",
    journal = "JHEP",
    volume = "04",
    pages = "009",
    year = "2023"
}

@article{Muller-Lennert:2013liu,
    author = {M{\"u}ller-Lennert, Martin and Dupuis, Fr{\'e}d{\'e}ric and Szehr, Oleg and Fehr, Serge and Tomamichel, Marco},
    title = "{On quantum R{\'e}nyi entropies: A new generalization and some properties}",
    eprint = "1306.3142",
    archivePrefix = "arXiv",
    primaryClass = "quant-ph",
    doi = "10.1063/1.4838856",
    journal = "J. Math. Phys.",
    volume = "54",
    number = "12",
    pages = "122203",
    year = "2013"
}

@article{Kato:2023aro,
  title={A remark on non-commutative {$L^p$}-spaces},
  author={Shinya Kato and Yoshimichi Ueda},
  eprint = {2307.01790},
  archivePrefix = "arXiv",
  primaryClass = "math.OA",
  journal={Studia Mathematica},
  year={2023},
  volume = {275},
  pages = {235-248},
  doi={10.4064/sm230724-11-10}
}

@article{Hiai:2024qve,
    author = "Hiai, Fumio and Jen{\v{c}}ov{\'a}, Anna",
    title = "{$\alpha $-z-R{\'e}nyi Divergences in von Neumann Algebras: Data Processing Inequality, Reversibility, and Monotonicity Properties in $\alpha ,z$}",
    eprint = "2404.07617",
    archivePrefix = "arXiv",
    primaryClass = "quant-ph",
    doi = "10.1007/s00220-024-05124-1",
    journal = "Commun. Math. Phys.",
    volume = "405",
    number = "11",
    pages = "271",
    year = "2024"
}

@article{Hollands:2025glm,
    author = "Hollands, Stefan and Longo, Roberto",
    title = "{A New Proof of the QNEC}",
    eprint = "2503.04651",
    archivePrefix = "arXiv",
    primaryClass = "hep-th",
    doi = "10.1007/s00220-025-05450-y",
    journal = "Commun. Math. Phys.",
    volume = "406",
    number = "11",
    pages = "269",
    year = "2025"
}

@article{Jencova:2017let,
    author = "Jencova, Anna",
    title = "{Preservation of a quantum Renyi relative entropy implies existence of a recovery map}",
    eprint = "1604.02831",
    archivePrefix = "arXiv",
    primaryClass = "quant-ph",
    doi = "10.1088/1751-8121/aa5661",
    journal = "J. Phys. A",
    volume = "50",
    pages = "085303",
    year = "2017"
}

@article{Witten:2021unn,
    author = "Witten, Edward",
    title = "{Gravity and the crossed product}",
    eprint = "2112.12828",
    archivePrefix = "arXiv",
    primaryClass = "hep-th",
    doi = "10.1007/JHEP10(2022)008",
    journal = "JHEP",
    volume = "10",
    pages = "008",
    year = "2022"
}

@article{Chandrasekaran:2026pnc,
    author = "Chandrasekaran, Venkatesa and Flanagan, {\'E}anna {\'E}.",
    title = "{Subregion algebras in classical and quantum gravity}",
    eprint = "2601.07915",
    archivePrefix = "arXiv",
    primaryClass = "hep-th",
    month = "1",
    year = "2026"
}

@article{Longo:2018obd,
    author = "Longo, Roberto",
    title = "{Entropy distribution of localised states}",
    eprint = "1809.03358",
    archivePrefix = "arXiv",
    primaryClass = "hep-th",
    doi = "10.1007/s00220-019-03332-8",
    journal = "Commun. Math. Phys.",
    volume = "373",
    number = "2",
    pages = "473--505",
    year = "2019"
}

@article{Kibe:2026wsg,
    author = "Kibe, Tanay and Roy, Pratik",
    title = "{A general proof of integer R{\'e}nyi QNEC}",
    eprint = "2605.15272",
    archivePrefix = "arXiv",
    primaryClass = "hep-th",
    month = "5",
    year = "2026"
}

@article{Mosonyi:2022nsp,
    author = "Mosonyi, Mil{\'a}n and Bunth, Gergely and Vrana, P{\'e}ter",
    title = "{Geometric relative entropies and barycentric R{\'e}nyi divergences}",
    eprint = "2207.14282",
    archivePrefix = "arXiv",
    primaryClass = "quant-ph",
    doi = "10.1016/j.laa.2024.06.005",
    journal = "Linear Algebra Appl.",
    volume = "699",
    pages = "159--276",
    year = "2024"
}

\end{document}